\documentclass[12pt,a4j]{article}
\usepackage{amsmath,amsthm,amssymb,bm,fancybox,multirow,hhline,psfrag}
\usepackage{url}
\usepackage[dvipdfm]{graphicx}
\usepackage{natbib}
\usepackage{algorithm,algorithmic}
\newtheorem{lemma}{Lemma}[section]
\newtheorem{remark}{Remark}[section]
\newtheorem{example}{Example}[section]

\def\Bka{{\it Biometrika}}
\def\AIC{AIC}
\def\BIC{BIC}
\def\GCV{GCV}
\def\CV{CV}
\def\DF{\mathrm{df}}
\def\MSE{MSE}

\def\SE{SE}
\def\SD{SD}
\def\ZZ{ZZ}
\def\NN{NN}

\def\T{{ \mathrm{\scriptscriptstyle T} }}

\def\cov{\mathrm{cov}}

\begin{document}
{
\begin{center}
\textbf{\Large Efficient algorithm to select tuning parameters in sparse regression modeling with regularization}
\end{center}
\begin{center}
\large {Kei Hirose$^1$, Shohei Tateishi$^2$ 
and Sadanori Konishi$^3$}
\end{center}

\begin{center}
{\it {\small
$^1$ Division of Mathematical Science, Graduate School of Engineering Science, Osaka University,\\
1-3, Machikaneyama-cho, Toyonaka, Osaka, 560-8531, Japan \\

$^2$ Toyama Chemical Co., Ltd., 3-2-5, Nishi-Shinjuku, Shinjuku-ku, Tokyo, 160-0023, Japan.\\

\vspace{1.2mm}

$^3$ Faculty of Science and Engineering, Chuo University,\\
 1-13-27 Kasuga, Bunkyo-ku, Tokyo, 112-8551, Japan. \\
}}

\vspace{1.5mm}

{\it {\small E-mail: mail@keihirose.com, \ shohei.tateishi@gmail.com,  
konishi@math.chuo-u.ac.jp.\\
}}

\end{center}
\begin{abstract}
In sparse regression modeling via regularization such as the lasso, it is important to select appropriate values of tuning parameters including regularization parameters.  The choice of tuning parameters can be viewed as a model selection and evaluation problem.  Mallows' $C_p$ type criteria may be used as a tuning parameter selection tool in lasso-type regularization methods, for which the concept of degrees of freedom plays a key role.  In the present paper, we propose an efficient algorithm that computes the degrees of freedom by extending the generalized path seeking algorithm.  Our procedure allows us to construct model selection criteria for evaluating models estimated by regularization with a wide variety of convex and non-convex penalties.  Monte Carlo simulations demonstrate that our methodology performs well in various situations.  A real data example is also given to illustrate our procedure.
 \end{abstract}


\section{Introduction}
Variable selection is fundamentally important in high-dimensional linear regression modeling.  
Traditional variable selection procedures follow the best subset selection along with  model selection criteria such as Akaike's information criterion \citep{Akaike:1973} and the Bayesian information criterion \citep{Schwarz:1978}.  However, the best subset selection  is  often unstable because of its inherent discreteness \citep{Breiman:1996}, and then the resulting model has poor prediction accuracy. To overcome this drawback of the subset selection,   \citet{Tibshirani:1996} proposed the lasso, which shrinks some coefficients toward exactly zero by imposing an $L_1$ penalty on regression coefficients, resulting in simultaneous model selection and estimation procedure. 

Over the past 15 years, there has been a considerable amount of lasso-type penalization methods in literature: bridge regression \citep{FrankFriedman:1993,Fu:1998}, smoothly clipped absolute deviation \citep{FanLi:2001}, elastic net \citep{ZouandHastie:2005},  group lasso \citep{YuanandLin:2006},  adaptive lasso \citep{Zou:2006},  composite absolute penalties family \citep{Zhaoetal:2009}, minimax concave penalty  \citep{Zhang:2010} and generalized elastic net \citep{Friedman2008} along with many other regularization techniques.  
It is well known that the solutions are not usually expressed in a closed form, since the penalty term includes non-differentiable function.   A number of  researchers have presented efficient algorithms to obtain the entire solutions 
(e.g., least angle regression, \citealp{Efronetal:2004}; 
 coordinate descent algorithm, \citealp{Friedmanetal:2007,Friedmanetal:2010}, \citealp{Mazumderetal:2009}; generalized path seeking, \citealp{Friedman2008}). 

A crucial issue in the sparse regression modeling via regularization is the selection of adjusted tuning parameters including regularization parameters, because the regularization parameters identify a set of non-zero coefficients and then assign a set of variables to be included in a model.   
Choosing the tuning parameters can be viewed as a model selection and evaluation problem.  
Mallows' $C_p$ type criteria \citep{Mallows:1973} estimate the prediction error of the fitted model, and  give better accuracy than cross validation in some situations \citep{Efron:2004}.  The concept of degrees of freedom (e.g., \citealp{Ye:1998,Efron:1986,Efron:2004}) plays a key role in the theory of $C_p$ type criteria.  

In a practical situation, however, it is difficult to directly derive an analytical expression of  (unbiased estimator of) degrees of freedom for sparse regression modeling.  A few researchers have derived the analytical results by using the Stein's unbiased risk estimator \citep{Stein:1981}  for only specific penalties.  \citet{Zouetal:2007} showed that the number of non-zero coefficients is an unbiased estimate of the degrees of freedom of the lasso.  \citet{Kato:2009} derived an unbiased estimate of the degrees freedom of the lasso, group lasso and fused lasso based on a differential geometric approach.  \citet{Mazumderetal:2009} proposed a re-parametrization of minimax concave penalty, which enables us to calibrate the degrees of freedom of minimax concave family.  However, these selection procedures do not cover more general regularization methods via convex and non-convex penalties.  
In such a situation, the cross validation and the bootstrap (e.g., \citealp{Ye:1998,Efron:2004,ShenYe:2002,Shenetal:2004}) may be useful to estimate the degrees of freedom. These approaches, however, can be computationally expensive, and often yield unstable estimates.   

In the present paper, we propose a new algorithm that can iteratively calculate the degrees of freedom by extending the generalized path seeking algorithm \citep{Friedman2008}.  
The proposed procedure can be applied to a wide variety of convex and non-convex penalties including the generalized elastic net family \citep{Friedman2008}.   
Furthermore, our algorithm is computationally-efficient, because there is no need to perform  numerical optimization to obtain the solutions and degrees of freedom at each step.  
The proposed methodology is investigated through the analysis of real data and Monte Carlo simulations.   Numerical results show that $C_p$ criterion based on our algorithm  performs well in various situations.    




The remainder of this paper is organized as follows: Section 2 briefly describes the degrees of freedom in linear regression models. In Section 3, we introduce a new algorithm that iteratively computes the degrees of freedom by extending the generalized path seeking.  
Section 4 presents numerical results for both artificial and real datasets. Some concluding remarks are given in Section 5.

\section{Degrees of freedom in linear regression models}
In linear regression models, the degrees of freedom can be used as a model complexity measure in Mallows' $C_p$ type criteria.  Suppose that  $ \bm{x}_j=(x_{1j},\dots,x_{Nj})^T$ ($j=1,\dots,p$) are predictors and $ \bm{y} = (y_1,\dots,y_N)^T$ is a response vector.  Without loss of generality, it  is assumed that the response is centered and the predictors are standardized by changing a location and employing scale transformations
\begin{equation*}
\sum_{i=1}^N y_i =0,\quad \sum_{i=1}^N x_{ij} =0, \quad  \sum_{i=1}^N x_{ij}^2 = 1 \quad  (j=1,\dots,p).
\end{equation*}
Consider the linear regression model
\begin{eqnarray*}
 \bm{y} = X \bm{\beta} +  \bm{\varepsilon},
\end{eqnarray*}
where  $X=( \bm{x}_1,\dots, \bm{x}_p)$ is an $N \times p$ predictor matrix, $ \bm{\beta} = (\beta_1,\dots,\beta_p)^T$ is a coefficient vector and $ \bm{\varepsilon} =(\varepsilon_1,\dots,\varepsilon_N)^T$ is an error vector with  $E[ \bm{\varepsilon}] =  {0}$ and $V[ \bm{\varepsilon}] = \sigma^2I$.  Here  $I$ is an identity matrix. 

The linear regression model is estimated by the penalized least square method
\begin{equation}
\hat{ \bm{\beta}}(t) =  \underset{\bm{\beta}}{\operatorname{argmin}}  \ R( \bm{\beta}) \quad {\rm s.t.} \quad   P( \bm{\beta}) \le t,  \label{settei}
\end{equation}
where $R( \bm{\beta})$ is a squared error loss function
\begin{equation}
R( \bm{\beta}) = ( \bm{y} - X  \bm{\beta})^T( \bm{y} - X  \bm{\beta}) / N, \label{squaredloss}
\end{equation}
$P( \bm{\beta})$ is a penalty term which yields sparse solutions (e.g., the lasso penalty is $P( \bm{\beta}) = \sum_j  | \beta_j |$), and $t$ is a tuning parameter.  An equivalent formulation of (\ref{settei}) is
\begin{equation*}
\hat{ \bm{\beta}}(\lambda) = \underset{\bm{\beta}}{\operatorname{argmin}}\{ R( \bm{\beta}) + \lambda P( \bm{\beta}) \},  \label{settei2}
\end{equation*}
where  $\lambda$ is a regularization parameter, which corresponds to $t$ in (\ref{settei}).  


We consider the problem of selecting an appropriate value of tuning parameter $t$ (or $\lambda$) by using $C_p$ type criteria, for which the concept of degrees of freedom plays a key role \citep{Ye:1998}.   Assume that the expectation and the variance-covariance matrix of the response vector $ \bm{y}$ are
\begin{equation}
E[ \bm{y}]=  \bm{\mu},\quad V( \bm{y}) = E[( \bm{y} -  \bm{\mu})( \bm{y} -  \bm{\mu})^T] = \tau^2 I,  \label{true_distribution_of_y}
\end{equation}
where $ \bm{\mu}$ is a true mean vector and ${\tau}^2$ is a true variance.  Given a modeling procedure $m$, the estimate $\hat{ \bm{\mu}}= m( \bm{y})$ can be produced from the data vector $\bm{y}$.  
Then, the degrees of freedom of the fitting procedure $m$ is defined as (\citealp{Ye:1998,Efron:1986,Efron:2004})
\begin{eqnarray}
\DF = \sum_{i=1}^N \frac{\mathrm{cov}(\hat{{\mu}}_i,y_i)}{\tau^2}, \label{gdf}
\end{eqnarray}
where $\hat{{\mu}}_i$ is the $i$th element of $\hat{\bm{\mu}}$. For example, when the estimator $\hat{ \bm{\mu}}$ is expressed as a linear combination of response vector, i.e. $\hat{ \bm{\mu}} = H \bm{y}$ with $H$ being independent of $ \bm{y}$, the degrees of freedom is $\mathrm{tr}(H)$.  The trace of matrix $H$ is referred to as an effective number of parameters \citep{HastieandTibshirani:1990}, which is widely used to select the tuning parameter in ridge-type regression.  
In sparse regression modeling such as the lasso, however,  it is difficult to derive the degrees of freedom, since the penalty term is not differentiable at $\bm{\beta}_j=\bm{0}$ ($j=1,\dots,p$) so that the solutions are not usually expressed in a closed form.  
 Mallows' $C_p$ criterion, which is an unbiased estimator of the true prediction error, can be constructed with the degrees of freedom defined in (\ref{gdf}).   
 Assume that the response vector $ \bm{y}$ is generated according to (\ref{true_distribution_of_y}), 
and the true expectation $\bm{\mu} $ is estimated by linear regression model.  As a criterion to measure the effectiveness of the model, we consider the expected error (e.g., \citealp{Hastieetal:2008}) defined by
\begin{equation}
{\rm Err} =E_{ \bm{y}} E_{ \bm{y}^{\rm new}}[(\hat{ \bm{\mu}} -  \bm{y}^{\rm new})^T(\hat{ \bm{\mu}} -  \bm{y}^{\rm new})], \label{expected error}
\end{equation}
where the expectation  $E_{ \bm{y}^{\rm new}}$ is taken over $ \bm{y}^{\rm new} \sim ( \bm{\mu},\tau^2I)$ independent of $ \bm{y}$.   
\begin{lemma}\label{Lem_cp}
The expected error in (\ref{expected error}) can be expressed as 
\begin{eqnarray}
{\rm Err} &=&  E_{ \bm{y}}\left[\| \bm{y}-\hat{ \bm{\mu}}\|^2+ 2\tau^2 \DF \right]. \label{Mallows_der}
\end{eqnarray}
\end{lemma}
\begin{proof}
The proof is in Appendix.   
\end{proof}
Lemma~\ref{Lem_cp} 
suggests $C_p$ criterion (e.g., \citealp{Efron:2004})
\begin{eqnarray*}
C_p= \| \bm{y}-\hat{ \bm{\mu}}\|^2+2\tau^2 \DF, \label{Cp_criterion}
\end{eqnarray*}
which is an unbiased estimator of the expected error in (\ref{expected error}). The optimal model is selected by minimizing $C_p$.  As an estimator of the true variance of  $\tau^2$, the unbiased estimator of error variance of the most  complex model is usually used.  


The degrees of freedom can lead to several model selection criteria, which are summarized in Table \ref{table: summary_penalty}.   \citet{Zouetal:2007} introduced Akaike's information criterion (\AIC; \ \citealp{Akaike:1973}) and Bayesian information criterion  (\BIC; \ \citealp{Schwarz:1978}). 
\citet{Wangetal:2007,Wangetal:2009} showed that the Bayesian information criterion holds the consistency in model selection.  We also introduce  bias corrected Akaike's information criterion (\AIC$_C$; \citealp{Sugiura:1978,Hurvichetal:1998}) and 
generalized cross validation (\GCV; \citealp{Craven:1979}).  
These two criteria do not need the true variance $\tau^2$. 
\begin{table}[t]
\caption{Summary of model selection criteria based on the degrees of freedom.}
\begin{center}
\begin{tabular}{ll}
\hline
\multicolumn{1}{c}{Criterion}&\multicolumn{1}{c}{Formula}\\
\hline
$C_p$ &${\| \bm{y} - \hat{ \bm{\mu}} \|^2}+2\tau^2 \DF$\\ \vspace{1mm}
\AIC &$ N\log (2\pi\tau^2) + \dfrac{\| \bm{y} - \hat{ \bm{\mu}} \|^2}{\tau^2} +  2 \DF$\\
\AIC$_C$ &$N \log \left(2\pi \dfrac{ \| \bm{y} - \hat{ \bm{\mu}} \|^2}{N} \right) + N -  \dfrac{2N {\DF}}{ N-{\DF}-1} $\\
\BIC &$ N\log (2\pi\tau^2) + \dfrac{\| \bm{y} - \hat{ \bm{\mu}} \|^2}{\tau^2} +  \log N \DF$\\
\GCV &$\dfrac{1}{N} \dfrac{\|  \bm{y} - \hat{ \bm{\mu}} \|^2}{ (1- \DF / N)^2}$\\
\hline
\end{tabular}
\end{center}
\label{table: summary_penalty}
\end{table}

\section{Efficient algorithm for computing the degrees of freedom} \label{GPS_SEC}
In this section,  first,  the generalized path seeking algorithm is briefly described.  Then, a new algorithm that iteratively computes the degrees of freedom  is introduced.   Furthermore, we modify the algorithm to ease the computational burden for large sample sizes.  
\subsection{Generalized path seeking algorithm}
 \citet{Friedman2008} proposed  the generalized path seeking, which is a fast algorithm to solve the problem (\ref{settei}).  The generalized path seeking can produce the entire solutions that closely approximate those for a wide variety of convex and non-convex constraints.  Suppose that the penalty term $P( \bm{\beta})$ satisfies following condition:
\begin{equation}
\left\{ \frac{\partial P( \bm{\beta})}{\partial |\beta_j|}>0 \mid j=1,\dots,p \right\}. \label{GPFcondition}
\end{equation}
This condition defines a class of penalties where each member in the class is a monotone increasing function of absolute value of each of its arguments. For example, the lasso penalty  $P( \bm{\beta}) = \sum_{j=1}^p |\beta_j|$  is included in this class, because $\partial P( \bm{\beta}) / \partial |\beta_j|=1>0 $.  Similarly, elastic net \citep{ZouandHastie:2005}, group lasso \citep{YuanandLin:2006}, adaptive lasso \citep{Zou:2006},  composite absolute penalties family \citep{Zhaoetal:2009}, minimax concave penalty  \citep{Zhang:2010} and generalized elastic net \citep{Friedman2008} with many other convex and non-convex penalties are included in this class.  

Denote $\hat{ \bm{\beta}}(t)$ is the solution at  tuning parameter $t$.   The generalized path seeking algorithm starts at $t=0$ with $\hat{ \bm{\beta}}(0)= {0}$. The solution can be iteratively computed:  for given $\hat{ \bm{\beta}}(t)$, the solution $\hat{ \bm{\beta}}(t + \Delta t )$ can be produced, where  $\Delta t$ is a small positive value.   
Suppose the path $\hat{ \bm{\beta}}(t)$ is a continuous function of $t$ and all coefficient paths $\{ \hat{\beta}_j(t)    \mid j=1,\dots,p \}$ are monotone function of $t$, that is,
$
\{ | \hat{\beta}_j(t+ \Delta t) | \ge | \hat{\beta}_j(t) |     \mid j=1,\dots,p \}. \label{monotone_condition}
$
For each step, one element of coefficient vector $\hat{ \bm{\beta}}(t)$, say $\hat{\beta}_k(t)$, is  incriminated in a correct direction $\lambda_k(t)$ with all other coefficients remaining unchanged, i.e.
\begin{eqnarray}
\hat{\beta}_k(t+\Delta t) &=& \hat{\beta}_k(t) + \Delta t \cdot  \lambda_k(t),\label{addall}\\
\{ \hat{\beta}_j(t+\Delta t) &=& \hat{\beta}_j(t)  \}_{j \ne k}, \label{addone}
\end{eqnarray}  
where $k$ and  $\lambda_k(t)$ are defined as
\begin{eqnarray}
k&=&\underset{j \in \{1,\dots,p\}} {\operatorname{argmax}} \ |g_j(t)|/p_j(t), \nonumber \\
\lambda_{k}(t) &=& g_k(t)/p_k(t). \label{update_beta_original}
\end{eqnarray}
Here $g_j(t) $ and $p_j(t) $  are
\begin{eqnarray*}
g_j(t) &=& -\left. \left[ \frac{\partial R( \bm{\beta})}{\partial \beta_j}  \right] \right|_{ \bm{\beta}=\hat{ \bm{\beta}}(t)}, \\
p_j(t) &=& \left.\left[ \frac{\partial P( \bm{\beta})}{\partial |\beta_j|}  \right] \right|_{ \bm{\beta}=\hat{ \bm{\beta}}(t)}.
\end{eqnarray*}
The derivation of the generalized path seeking algorithm is in Appendix.
\begin{remark}
We assumed that $\hat{ \bm{\beta}}(t)$ is continuous and each element is monotone function of $t$.  Although these conditions can be satisfied in most cases,  sometimes $\hat{ \bm{\beta}}(t)$ is  discontinuous or non-monotone function.  \citet{Friedman2008} proposed an approach for non-monotone case, which is as follows: first, we define a set $S=\{ j \mid \lambda_j(t) \cdot \hat{\beta}_j(t) < 0 \}$.  When $S$ is not empty,  the index $k$ is selected by $k = \mathrm{argmax}_{j \in S} |\lambda_j(t)|$.  Otherwise, $k = \mathrm{argmax}_{j \in \{1,\dots,p\}}  |\lambda_j(t)|$.   The detailed description of  discontinuous case is also given in \citet{Friedman2008}.   
 \end{remark} 


When $p_j(t)=1$ (i.e. the lasso penalty), (\ref{addall}) yields
\begin{equation}
\hat{\beta}_{k}(t+\Delta t)=\hat{\beta}_{k}(t) + \Delta t  \cdot g_{k}(t). \label{update_beta}
\end{equation}
Note that the updated coefficient in (\ref{addall}) and (\ref{update_beta}) moves in  the same direction even if the lasso penalty is not applied, because the condition in (\ref{GPFcondition}) yields $sign(g_k(t)) = sign(\lambda_k(t))$.  This means that 
the update equations (\ref{addall}) and (\ref{update_beta}) produce the same solution path when $\Delta t \rightarrow 0$ unless $p_j(t)$ or $1/p_j(t)$ diverges. Therefore, we can use the update equation in (\ref{update_beta}) instead of  (\ref{addall}).   If the update equation in (\ref{update_beta}) is applied, an iterative algorithm that computes the degrees of freedom in (\ref{gdf}) can be derived.

From (\ref{addone}) and  (\ref{update_beta}), the  predicted value at $t+\Delta t$ is  
\begin{eqnarray}
\hat{ \bm{\mu}}(t+\Delta t) 
&=&\hat{ \bm{\mu}}(t)  +  \Delta t \cdot g_{k}(t)  \bm{x}_{k} . \label{update_coefficient}
\end{eqnarray}
Because the loss function $R( \bm{\beta})$ is squared loss as (\ref{squaredloss}),  we have $g_{k}(t)= 2 \bm{x}_{k}^T( \bm{y} - \hat{ \bm{\mu}}(t))/N$. Thus,
$\hat{ \bm{\mu}}(t+\Delta t)$ is
\begin{eqnarray}
\hat{ \bm{\mu}}(t+\Delta t) 
&=& \hat{ \bm{\mu}}(t)  + \frac{2}{N} \Delta t \  \bm{x}_{k}  \bm{x}_{k}^T  \cdot ( \bm{y} - \hat{ \bm{\mu}}(t) )  . \label{update_mu}
\end{eqnarray}
\begin{example}\label{example:orthogonalGPS}
Let $X$ be orthogonal, i.e. $X^TX=I$.  
By substituting (\ref{update_mu}) into $g_{j}(t)= 2 \bm{x}_{j}^T( \bm{y} - \hat{ \bm{\mu}}(t))/N$,  the update equation of $g_j(t)$ is
\begin{equation*}
g_j(t+\Delta t) = \left\{
\begin{array}{cc}
\left(1-{2\Delta t}/{N}\right)g_j(t) & \quad (j=k),\\
g_j(t)& \quad (j \ne  k).\\
\end{array}
\right.
\end{equation*}
Because of the orthogonality, we have 
\begin{equation*}
g_j(t) = (1-2\Delta  t/N)^{t_j} \cdot 2 \bm{x}_j^T \bm{y} /N,
\end{equation*}
where $t_j$ is the number of times that $j$th coefficient is updated until time step $t$.  It is shown that the absolute value of $g_j(t)$ is monotone non-increasing function and $g_j(t) \rightarrow 0$ when $t_j \rightarrow \infty$.  When $t \rightarrow \infty$,  the least squared estimates can be obtained because $g_j(t) \rightarrow 0$ for all $j=1,\dots,p$.
\end{example}

\subsection{Derivation of update equation of degrees of freedom} \label{derivation of DF}
Equation (\ref{update_mu}) suggests the update equation of the covariance matrix in  (\ref{gdf}) as follows:
\begin{eqnarray}
 \frac{\cov(\hat{ \bm{\mu}}(t+\Delta t), \bm{y})}{\tau^2}&=& \frac{\cov(\hat{ \bm{\mu}}(t), \bm{y})}{\tau^2} + \frac{2}{N} \Delta t \  \bm{x}_{k}  \bm{x}_{k}^T  \left\{ I  -  \frac{\cov(\hat{ \bm{\mu}}(t), \bm{y})}{\tau^2} \right\}.  \label{gpsdf_lasso}
\end{eqnarray}
The degrees of freedom is iteratively calculated by taking the trace of  (\ref{gpsdf_lasso}).   The initial value of ${\cov(\hat{ \bm{\mu}}(t), \bm{y})} / {\tau^2}$ is set to zero-matrix $ {O}$ because of the following equation:
\begin{equation*}
 \frac{\mathrm{cov}(\hat{ \bm{\mu}}(0) , \bm{y}) }{\tau^2} =  \frac{\mathrm{cov}( {0}, \bm{y}) }{\tau^2} = {O}.
  \end{equation*}

Let 
$M(t) = {\cov(\hat{ \bm{\mu}}(t), \bm{y})} / {\tau^2}$  
and  $k(t)$ be the index of updated element of coefficient vector at time step $t$.  The update equation of the degrees of freedom in (\ref{gpsdf_lasso}) can be expressed as
\begin{equation}
I-M(t+\Delta t)=(I-\alpha   \bm{x}_{k(t)}  \bm{x}_{k(t)}^T)(I-M(t))  ,  \label{gpsdf_lasso_adj0}
  \end{equation}
  where  $\alpha = 2\Delta t /N$.   Then, the covariance matrix can be updated by
  \begin{equation}
M(t)=I-(I-\alpha   \bm{x}_{k(t-1)}  \bm{x}_{k(t-1)}^T ) (I-\alpha   \bm{x}_{k(t-2)}  \bm{x}_{k(t-2)}^T ) \cdots (I-\alpha   \bm{x}_{k(1)}  \bm{x}_{k(1)}^T ).    \label{gpsdf_lasso_adj}
  \end{equation}
  
\begin{example} \label{example1}
The degrees of freedom can be easily derived when $X$ is orthogonal.  Because of the orthogonality,  the covariance matrix in (\ref{gpsdf_lasso_adj}) can be calculated as
  \begin{eqnarray*}
M(t)&=&I-(I-\alpha   \bm{x}_1  \bm{x}_1^T )^{t_1}(I-\alpha   \bm{x}_2  \bm{x}_2^T )^{t_2} \cdots (I-\alpha   \bm{x}_p  \bm{x}_p^T )^{t_p}\\
&=& \sum_{j=1}^p \{1-(1-\alpha)^{t_j}   \} \bm{x}_j \bm{x}_j^T,   \label{gpsdf_lasso_adj_ex}
  \end{eqnarray*}
where $t_j$ is defined in Example \ref{example:orthogonalGPS}.  Then, the degrees of freedom is
  \begin{equation*}
{\rm tr} \{ M(t) \}=\sum_{j=1}^p \{1-(1-\alpha)^{t_j}   \}.\label{gpsdf_lasso_adj_ex2}
  \end{equation*}
 When $t$ is very small, the degrees of freedom is close to $0$ since $\alpha = 2\Delta t/N$ is sufficiently small.  As $t$ gets larger, the degrees of freedom increases since $(1-\alpha)^{t_j} > (1-\alpha)^{t_j+1}$.   
   When $t_j \rightarrow \infty$ for all $j$, the degrees of freedom becomes the number of parameters, which coincides with the degrees of freedom of least squared estimates.
   \end{example}  


\subsection{Modification of the update equation} \label{Modification of the algorithm}
The update equation in (\ref{update_coefficient}) causes little change in predicted values from $t$ to $t+ \Delta t$ near the least squared estimates, because $|g_k(t) | =|2 \bm{x}_k^T( \bm{y}-\hat{ \bm{\mu}}(t))/N|$ is very close to zero.  In order to overcome this difficulty, we update the $k(t)$th element of coefficient vector $m$ times. Here $m$ is an integer which becomes large near least squared estimates.    Since  $g_k(t+\Delta t) = \left(1-{2\Delta t}/{N}\right)g_k(t) $ as shown in the Example \ref{example:orthogonalGPS},  the $\hat{\beta}_{k}(t+m\Delta t)$ is
\begin{eqnarray}
\hat{\beta}_{k}(t+m\Delta t)=\hat{\beta}_{k}(t) +  \frac{1-(1-\alpha)^{m}}{\alpha} \Delta t \cdot g_k(t). \label{GPS_update_modification2}
\end{eqnarray}
The update equation in (\ref{GPS_update_modification2}) can be applied even when $m$ is a positive real value.  

The following update equation can be used so that the coefficient is appropriately updated near the least squared estimates:
\begin{eqnarray}
\hat{\beta}_{k}(t+m\Delta t)&=&\hat{\beta}_{k}(t) +  \Delta t \cdot sign( g_k(t)) \nonumber \\
&=&\hat{\beta}_{k}(t) +  \frac{1}{|g_k(t)|}\Delta t \cdot g_k(t).
 \label{GPS_update_modification}
\end{eqnarray}
Equations (\ref{GPS_update_modification2}) and (\ref{GPS_update_modification}) give us
\begin{eqnarray}
m = \frac{\log(1-\alpha/|g_k(t)|)}{\log (1-\alpha)}. \label{GPS_update_modification3}
\end{eqnarray}
It should be assumed that $\alpha < |g_k(t)|$ for any step so that $\log(1-\alpha/|g_k(t)|)$ exists.  

When $m$ is given by (\ref{GPS_update_modification3}), the update equation of the degrees of freedom in  (\ref{gpsdf_lasso_adj0}) can be replaced with 
\begin{eqnarray}
M(t+m\Delta t)&=&I-\left\{ I-   \alpha_t  \bm{x}_{k(t)}  \bm{x}_{k(t)}^T \right\} (I-M(t)) ,  \label{naive update}
\end{eqnarray}
where $\alpha_t = {\alpha } / { |g_k(t)|}$ .  
The algorithm that computes the solutions and the degrees of freedom is given in Algorithm \ref{GPSalgorithm for GDF modified}.

\begin{algorithm}[t]
\caption{An iterative algorithm that computes the solution and the degrees of freedom.}
\label{GPSalgorithm for GDF modified}
\begin{algorithmic}[1]
\STATE $t=0$.
\WHILE{$ \{ |g_j(t)| > \alpha \}$  $(j=1,\dots,p)$}
\STATE Compute $\{ g_j(t) \}$ and  $\{ \lambda_j(t) \}$ $(j=1,\dots,p)$.
\STATE $S=\{ j \mid \lambda_j(t) \cdot \hat{\beta}_j(t) < 0  \}$ \label{GPS:S}
\IF{$S={\rm empty}$}
\STATE $k = \underset{j \in \{ 1,\dots,p \}}{\operatorname{argmax}}  |\lambda_j(t)|$
\ELSE
\STATE $k = \underset{j \in S}{\operatorname{argmax}}  |\lambda_j(t)|$
\ENDIF
\STATE Compute $m = {\log(1-\alpha/|g_k(t)|)} / {\log (1-\alpha)}$.
\STATE $\hat{\beta}_{k}(t+m\Delta t)=\hat{\beta}_{k}(t) + \Delta t \cdot sign(\lambda_{k}(t)) $ \label{GPSupdate}
\STATE  $\{ \hat{\beta}_{j}(t+m\Delta t)=\hat{\beta}_{j}(t)  \}_{j \ne k}$
\STATE Compute
\begin{equation*}
M(t+m\Delta t)=I-\left\{ I-  \alpha_t \bm{x}_{k(t)}  \bm{x}_{k(t)}^T \right\} (I-M(t)) . 
\end{equation*}  \label{df_cov_update}
\STATE  Compute $\DF(t+m\Delta t)=\mathrm{tr}\left\{M(t+m\Delta t)\right\} $
\STATE  $t \leftarrow t + m\Delta t$
\ENDWHILE
\end{algorithmic}
\end{algorithm}
\subsection{More efficient algorithm} \label{subsec: efficient}
The update equation (\ref{naive update}) suggests each step costs $O(N^2)$ operations to update the covariance matrix $M(t)$.  Because the number of iterations denoted by $T$ is usually very large such as $T=100000$, the proposed algorithm seems to be inefficient when $N$ is large.  However, a simple modification of the algorithm eases the computational burden. With the modified process, each step costs only $O(q^2)$ operations, where $q$ is the number of selected variables through the generalized path seeking algorithm:  $p-q$ variables are not selected at all steps for generalized path seeking algorithm.  
When $p$ is very large, $q$ is smaller than $p$ and early stopping is used.  For example, suppose that $p=5000$; if we do not want more than 200 variables in the final model, we set $q=200$ and stop the algorithm when 200 variables are selected.  

The modified algorithm is as follows:  first, the generalized path seeking algorithm is implemented to obtain the entire solutions.  The degrees of freedom is not computed, whereas the value of $g_k(t) $ $(t=1,\dots,T)$  should be stored in the memory.  Then, the QR decomposition of $N \times q$ matrix $X^*=( \bm{x}_{j_1} \cdots  \bm{x}_{j_q})$ is implemented, where $ \bm{x}_{j_1},\dots, \bm{x}_{j_q}$ are variables selected by the generalized path seeking algorithm.  Note that $\# \{j_1,\dots,j_q  \} = q$.  
The  matrix $X^*$ can be written as $X^* = QR$, where $Q$ is an $N \times q$ orthogonal matrix and $R$ is a  $q \times q$ upper triangular matrix.  Note that $ \bm{x}_{k}$ can be written as $Q \bm{r}_{k}$, where $ \bm{r}_{k}$ is the $q$-vector which consists of $k$th column of $R$. 
The update equation of the degrees of freedom based on  
(\ref{naive update}) is 
\begin{eqnarray}
\mathrm{tr} \{ M(t+m\Delta t) \}  &=& \mathrm{tr} \{Q^T M(t+m\Delta t) Q \}\cr
&=& q- \mathrm{tr} \{ (I-\alpha_{t}   \bm{r}_{k(t)}  \bm{r}_{k(t)}^T ) (I-\alpha_{t-1}  \bm{r}_{k(t-1)}  \bm{r}_{k(t-1)}^T  ) \cdots (I-\alpha_{1}   \bm{r}_{k(1)}  \bm{r}_{k(1)}^T  ) \} . \cr && \label{modified update equation}
\end{eqnarray}
Therefore, the computational cost of (\ref{modified update equation}) is only $O(q^2)$.  

We provide a package {\tt msgps} (Model Selection criteria via extension of Generalized Path Seeking), which computes Mallows' $C_p$ criterion, Akaike's information criterion, Bayesian information criterion  and generalized cross validation via the degrees of freedom given in Table \ref{table: summary_penalty}.  The package is implemented in the {\tt R} programming system \citep{R:2010}, and available from Comprehensive R Archive Network (CRAN) at \url{http://cran.r-project.org/web/packages/msgps/index.html}. 

\section{Numerical Examples}
\subsection{Monte Carlo simulations}
Monte Carlo simulations were conducted to investigate the effectiveness of our algorithm. The predictor vectors were generated from Gaussian distribution with mean vector zero.  The outcome values $y$ were generated by 
\begin{equation*}
y =  \bm{\beta}^T \bm{x} + \bm{\varepsilon}, \quad \bm{\varepsilon} \sim N({0}, \sigma^2).
\end{equation*}
The following four Examples are presented here.  
\begin{enumerate}
\item In Example 1, 200 data sets were generated with $N=20$  observations and eight predictors.  The true parameter was $ \bm{\beta} = (3,1.5,0,0,2,0,0)^T$ and $\sigma=3$.  The pairwise correlation between $ \bm{x}_i$ and $ \bm{x}_j$ was ${\rm cor}(i,j) = 0.5^{|i-j|}$.  
\item  Example 2 was the dense case.  The model was same as Example 1, but with $\beta_j=0.85$ $(j=1,\dots,8)$, and $\sigma=3$.  
\item The third example was same as Example 1, but with $ \bm{\beta} = (5,0,0,0,0,0,0,0)^T$ and $\sigma=2$.  In this model, the true $ \bm{\beta}$ is sparse.
\item In Example 4, a relatively large problem was considered.  200 data sets were generated with $N=100$ observations and 40 predictors.   We set
\begin{equation*}
 \bm{\beta} = (\underbrace{0,\dots,0}_{10},\underbrace{2,\dots,2}_{10},\underbrace{0,\dots,0}_{10},\underbrace{2,\dots,2}_{10})^T  
\end{equation*}
and $\sigma=15$. The pairwise correlation between $ \bm{x}_i$ and $ \bm{x}_j$ was ${\rm cor}(i,j) = 0.5$ ($i \ne j$). 
\end{enumerate} 
In this simulation study, there are 3 purposes as follows:
\begin{itemize}
\item Degrees of freedom: we investigated whether the proposed procedure can select adjusted tuning parameters compared with the degrees of freedom of the lasso given by \citet{Zouetal:2007}.
\item Model selection criteria for several penalties:  the performance of model selection criteria given in Table \ref{table: summary_penalty} was compared for the lasso, elastic net and generalized elastic net family \citep{Friedman2008}.
\item Speed:  the computational time based on (\ref{naive update}) was compared with that based on (\ref{modified update equation}).
\end{itemize}
A detailed description of each is presented.
\subsection*{Degrees of freedom} \label{sec: dfexample}

We compared the degrees of freedom computed by our procedure ($\DF_{gps}$, where $gps$ means generalized path seeking) with the degrees of freedom of the lasso proposed by \citet{Zouetal:2007} ($\DF_{zou}$).  The degrees of freedom of the lasso is   the number of non-zero coefficients.  Our method and  Zou's {et al}.  (2007) procedure do not yield identical result, since the $\DF_{zou}$ is an unbiased estimate of the degrees of freedom while $\DF_{gps}$ is the exact value of the degrees of freedom.  In this simulation study, the true value of $\tau^2$ was used to compute the model selection criteria.  


Table \ref{result_suchijikken_GPS1} shows the result of mean squared error (\MSE) and the standard deviation (\SD), which are   the mean and standard deviation of the following squared error (\SE):
\begin{eqnarray*}
{\SE}(s)  = \frac{1}{N} \| \hat{ \bm{\mu}}^{(s)} - X \bm{\beta}\|^2
\end{eqnarray*}
where $\hat{ \bm{\mu}}^{(s)} $ is the estimate of predicted values for $s$th dataset. The proportion of cases where zero (non-zero) coefficients correctly set to zero (non-zero), say, \ZZ\ (\NN), was also computed.  
We can see that 
\begin{itemize}
\item Our procedure slightly outperformed Zou's {et al.} (2007) one  in terms of minimizing the mean squared error for all examples.
\item Zou's {et al.} (2007) procedure selected zero coefficients correctly than the proposed procedure, while  our method correctly detected non-zero coefficients compared with  the  Zou's {et al.} (2007) approach.  This means our procedure tends to incorporate many more variables than the Zou's {et al.} (2007) one.
\end{itemize}


\begin{table}[t]
\caption{Mean squared error (\MSE), the standard deviation (\SD) and  the percentage of cases where zero (non-zero) coefficients correctly set to zero (non-zero), say, \ZZ\ (\NN), for our proposed procedure ($\DF_{gps}$) and \citet{Zouetal:2007} ($\DF_{zou}$). } \label{result_suchijikken_GPS1}
\begin{center}
\begin{tabular}{rrrrrrrrrr}
  \hline
&&\multicolumn{2}{c}{Ex. 1}&\multicolumn{2}{c}{Ex. 2}&\multicolumn{2}{c}{Ex. 3}&\multicolumn{2}{c}{Ex. 4}\\
&&$\DF_{gps}$&$\DF_{zou}$&$\DF_{gps}$&$\DF_{zou}$&$\DF_{gps}$&$\DF_{zou}$&$\DF_{gps}$&$\DF_{zou}$\\  \hline
& \MSE &  2.498 & 2.732 & 2.761 & 3.202 & 0.759 & 0.790& 41.35 & 42.37\\
&  \SD & 1.468 & 1.726 &1.353 & 1.727&  0.577 & 0.647& 10.67 & 12.36\\
&   \ZZ & 0.592 & 0.667 &--- & --- & 0.607 & 0.767  & 0.586 & 0.636 \\ 
&   \NN & 0.925 & 0.892 & 0.706 & 0.649 &1.000 & 1.000& 0.689 & 0.653 \\
 \hline
\end{tabular}
\end{center}
\end{table}

\subsection*{Model selection criteria for several penalties} \label{subsec:NEMSC}
We compared the performance of model selection criteria based on the degrees of freedom: $C_p$ criterion ($C_p$), bias corrected Akaike's information criterion (\AIC$_C$), Bayesian information criterion (\BIC) and generalized cross validation (\GCV).  
Because $C_p$ criterion and  Akaike's information criterion (\AIC) yield the same results when true error variance $\tau^2$ is given, the result of Akaike's information criterion is not presented in this paper. For $C_p$ criterion and Bayesian information criterion, we need to estimate the true error variance $\tau^2$.  The value of $\tau^2$ was estimated by the ordinary least squares  of most complex model.  The cross validation, which is one of the most popular methods to select the tuning parameter in sparse regression via regularization, was also applied.  Since the leave-one-out cross validation is computationally expensive, the 10-fold cross validation was used.

In this simulation study, we also compared the feature and performance of several penalties including the lasso, elastic net and generalized elastic net.   The elastic net and generalized elastic net are given  as follows:
\begin{enumerate}
\item Elastic net: 
\begin{equation*}
P( \bm{\beta}) = \sum_{j=1}^p \left\{ \frac{1}{2} \alpha  \beta_j^2 + (1-\alpha) |\beta_j| \right\}, \quad 0 \le \alpha \le 1.
\end{equation*}
Here $\alpha$ is a tuning parameter. Note that $\alpha=0$ yields the lasso, and $\alpha=1$ produces the ridge penalty.  

\item Generalized elastic net: 
\begin{equation*}
P( \bm{\beta}) = \sum_{j=1}^p \log \{ \alpha + (1-\alpha)|\beta_j|  \} , \quad 0 < \alpha < 1,
\end{equation*}
where $\alpha$ is the tuning parameter.  \citet{Friedman2008} showed the generalized elastic net approximates the power family penalties   $P( \bm{\beta}) = \sum |\beta_j|^{\gamma} \  (0 <  \gamma < 1)$, whereas the difference occurs at very small absolute coefficients.  The detailed description of the generalized elastic net is in \citet{Friedman2008}.  A similar idea of generalized elastic net is in \citet{Candes:2008}.

\end{enumerate}
Note that the degrees of freedom of the lasso (Zou's {et al}., 2007) cannot be directly applied to the generalized elastic net family.   

We computed  mean squared error (\MSE), standard deviation (\SD), the proportion of cases where zero (non-zero) coefficients correctly set to zero (non-zero), say, \ZZ\ (\NN).  
Tables \ref{result_suchijikken_GPS2}, \ref{result_suchijikken_GPS2_2} and \ref{result_suchijikken_GPS2_3} show the comparison of model selection criteria for the lasso, elastic net ($\alpha=0.5$) and generalized elastic net ($\alpha=0.5$).  
The detailed discussion of each example is as follows:
\begin{itemize}
\item  The generalized elastic net yielded the sparsest solution, while the elastic net produced the densest one. In Example 2, i.e. the dense case, the elastic net performed very well. 
On the other hand, in sparse case (Example 3), the generalized elastic net most often selected zero coefficients correctly.  
\item In most cases,  bias corrected Akaike's information criterion (\AIC$_C$) resulted in good performance in  terms of mean squared error.  The Bayesian information criterion (\BIC)  performed very  well in some cases (e.g., Examples 3 and 4 on the lasso).  In dense cases (Example 2), however,  the performance of Bayesian information criterion (\BIC) was poor.   
\item The performance of cross validation (\CV) for generalized elastic net family was excellent on Example 3.  On Examples 1, 2 and 4, however, the mean squared error was large compared with other model selection criteria based on the degrees of freedom. The cross validation estimates expected error by separating the  training data from the test data.   
Unfortunately,  the regularization method with non-convex penalty such as generalized elastic net does not produce unique solution: small change in the training data can result in different solution.   Thus, the cross validation may be unstable for non-convex penalty in many cases. 

\end{itemize}

\begin{table}[t]
\caption{Comparison of model selection criteria for the lasso. } \label{result_suchijikken_GPS2}
\begin{center}
\begin{tabular}{rrrrrrrrrrr}
  \hline
 && $C_p$ & \AIC$_C$ & \GCV & \BIC   & \CV \\ 
  \hline
Ex. 1& \MSE& 2.604 & 2.497 & 2.614 & 2.567 & 2.905  \\ 
 &  \SD &1.562 & 1.463 & 1.583 & 1.500 &  1.722 \\ 
 &  \ZZ &  0.519 & 0.562 & 0.498 & 0.622 & 0.605 \\ 
 &  \NN & 0.925 & 0.923 & 0.935 & 0.918 & 0.903 \\ 
   \cline{2-7}
Ex. 2 &  \MSE &2.807 & 2.772 & 2.781 & 2.891 &3.195  \\ 
  & \SD &  1.435 & 1.399 & 1.420 & 1.462 &  1.712 \\ 
 &  \ZZ & --- & --- & --- & --- & --- \\ 
  & \NN & 0.733 & 0.729 & 0.750 & 0.678 & 0.686  \\ 
 \cline{2-7}
Ex. 3 & \MSE &  0.855 & 0.790 & 0.879 & 0.744 & 0.986 \\ 
 &  \SD &  0.666 & 0.601 & 0.673 & 0.594 & 0.692 \\ 
&   \ZZ & 0.543 & 0.573 & 0.524 & 0.639 & 0.716 \\ 
&   \NN & 1.000 & 1.000 & 1.000 & 1.000 & 1.000  \\
 \cline{2-7}
Ex. 4 & \MSE &  41.66 & 42.91 & 43.82 & 39.22 & 43.59 \\ 
&   \SD & 10.72 & 11.19 & 11.74 & 9.512 & 11.83 \\ 
&   \ZZ & 0.577 & 0.545 & 0.530 & 0.671 & 0.643 \\ 
&   \NN & 0.692 & 0.704 & 0.713 & 0.653 & 0.655 \\ 
     \hline
\end{tabular}
\end{center}
\end{table}
\begin{table}[t]
\caption{Comparison of model selection criteria for the  elastic net ($\alpha=0.5$). } \label{result_suchijikken_GPS2_2}
\begin{center}
\begin{tabular}{rrrrrrrrrrr}
  \hline
 && $C_p$ & \AIC$_C$ & \GCV & \BIC   & \CV \\ 
  \hline
Ex. 1& \MSE & 2.860 & 2.814 & 2.826 & 2.933 & 3.017 \\ 
&   \SD &1.520 & 1.487 & 1.533 & 1.552 & 1.556 \\ 
&   \ZZ & 0.066 & 0.088 & 0.061 & 0.101 & 0.074 \\ 
&   \NN & 0.995 & 0.993 & 0.995 & 0.993 & 0.997 \\ 
     \cline{2-7}
Ex. 2 &  \MSE & 2.199 & 2.061 & 2.169 & 2.218 & 2.301 \\ 
&   \SD &1.434 & 1.303 & 1.352 & 1.489 & 1.432 \\ 
 &  \ZZ & --- & --- & --- & --- & --- \\ 
&   \NN &  0.973 & 0.967 & 0.976 & 0.958 & 0.964 \\ 
  \cline{2-7}
Ex. 3 & \MSE &1.566 & 1.675 & 1.577 & 1.714 & 1.720 \\ 
&   \SD &  0.754 & 0.834 & 0.755 & 0.878 & 0.960 \\ 
&   \ZZ &  0.014 & 0.031 & 0.016 & 0.036 & 0.028 \\ 
&   \NN & 1.000 & 1.000 & 1.000 & 1.000 & 1.000 \\ 
   \cline{2-7}
Ex. 4 & \MSE & 25.15 & 23.31 & 25.55 & 24.73 & 24.73 \\
 &  \SD & 10.51 & 8.656 & 11.27 & 9.031 & 9.431 \\ 
 &  \ZZ & 0.000 & 0.000 & 0.000 & 0.000 & 0.000 \\ 
 &  \NN & 1.000 & 1.000 & 1.000 & 1.000 & 1.000 \\ 
    \hline
\end{tabular}
\end{center}
\end{table}
\begin{table}[t]
\caption{Comparison of model selection criteria for the generalized  elastic net ($\alpha=0.5$). } \label{result_suchijikken_GPS2_3}
\begin{center}
\begin{tabular}{rrrrrrrrrrr}
  \hline
 && $C_p$ & \AIC$_C$ & \GCV & \BIC   & \CV \\ 
  \hline
Ex. 1& \MSE &  3.060 & 2.996 & 3.051 & 3.112 & 3.785 \\ 
&   \SD & 1.825 & 1.810 & 1.835 & 1.821 & 2.015 \\ 
&   \ZZ &  0.709 & 0.776 & 0.695 & 0.811 & 0.792 \\ 
&   \NN & 0.798 & 0.778 & 0.813 & 0.752 & 0.707 \\ 
   \cline{2-7}
Ex. 2 & \MSE&3.757 & 3.747 & 3.658 & 4.018 & 4.411 \\ 
  & \SD & 1.638 & 1.587 & 1.619 & 1.671 & 2.108 \\ 
 &  \ZZ & --- & --- & --- & --- & --- \\ 
  & \NN & 0.502 & 0.462 & 0.517 & 0.416 & 0.457 \\ 
  \cline{2-7}
Ex. 3 & \MSE & 0.889 & 0.803 & 0.949 & 0.683 & 0.485 \\ 
 &  \SD &  0.773 & 0.722 & 0.780 & 0.690 & 0.641 \\ 
  &  \ZZ &0.714 & 0.749 & 0.679 & 0.805 & 0.859 \\ 
 &  \NN & 1.000 & 1.000 & 1.000 & 1.000 & 1.000 \\ 
   \cline{2-7}
Ex. 4 & \MSE &74.68 & 74.12 & 75.15 & 78.75 & 80.56 \\ 
&   \SD &  18.22 & 17.97 & 18.90 & 16.19 & 22.41 \\ 
&   \ZZ & 0.796 & 0.796 & 0.766 & 0.912 & 0.815 \\ 
&   \NN & 0.379 & 0.382 & 0.408 & 0.262 & 0.352 \\ 
\hline
\end{tabular}
\end{center}
\end{table}

\clearpage
\subsection*{Speed}
The computational time based on update equation (\ref{naive update}) (na\"ive update) was compared with that based on (\ref{modified update equation}) (modified update).   In order to compare the timings for various number of samples, we changed the number of samples $N$ for all Examples: $N=100$, $200$ and $500$.  Table \ref{table:speed} shows the result of timings averaged over 200 runs for lasso penalty. All timings were carried out on an Intel Core 2 Duo 2.0 GH processor on Mac OS X.  Note that the ``timing" means the computational time of producing the solutions and computing the model selection criteria.   The speed based on na\"ive update was very slow when $N=500$, because we need $O(500^2)$ operations to compute the degree of freedom for each step. However, the modified algorithm was fast even when the number of samples was large. 

\begin{table}[t]
\caption{Computational time (seconds) based on update equation (\ref{naive update}) (na\"ive)  and (\ref{modified update equation})  (modified)  averaged over 200 runs for the lasso.}
\label{table:speed}
\begin{center}
\begin{tabular}{ccccccccccc}
  \hline
&\multicolumn{2}{c}{Ex. 1} &\multicolumn{2}{c}{Ex. 2} &\multicolumn{2}{c}{Ex. 3} &\multicolumn{2}{c}{Ex. 4}  \\
&na\"ive &modified&na\"ive &modified&na\"ive &modified&na\"ive &modified\\ 
  \hline
$N=100$ &1.399 &0.145 &1.434 & 0.164&1.398&0.143&1.521&0.343\\ 
$N=200$ & 4.827&  0.190  &4.834&  0.246& 4.820&0.188&5.012&0.375\\ 
$N=500$  & 69.92 &  0.313& 69.96& 0.351&69.89&0.310 &70.38&0.459\\ 
   \hline
\end{tabular}
\end{center}
\end{table}

\clearpage
\subsection{Application to diabetes data}
The proposed algorithm was applied to diabetes data \citep{Efronetal:2004}, which has $N= 442$ and $p = 10$.  Ten baseline predictors include age, sex, body mass index (bmi), average blood pressure (bp), and six blood serum  measurements (tc, ldl, hdl, tch, ltg, glu).  The response is a quantitative measure of disease progression one year after baseline.  We considered the following three penalties:  the lasso, elastic net ($\alpha=0.5$), and generalized elastic net ($\alpha=0.5$).  The entire solution path along with the solution selected by $C_p$ criterion, and the degrees of freedom  are presented in Figure \ref{fig:path_df_diabetes}.  
On the lasso penalty, the degrees of freedom of the lasso \citep{Zouetal:2007} is also depicted.   The degrees of freedom of our procedure was smaller than that of the lasso \citep{Zouetal:2007} except for $\|\hat{ \bm{\beta}}(t)\|  \in [2838, 2851]$, where the degrees of freedom of the lasso decreased because the non-zero coefficient of $7$th variable became zero at $t=2838$. 

On the elastic net penalty with $\alpha=0.5$, 
the degrees of freedom  increased rapidly at some points.  For example, when $\|\hat{ \bm{\beta}}(t)\|=343.76$,  
the degrees of freedom increased by about 0.864.  
Figure \ref{fig:2nd_df_diabetes} shows the solution path of 2nd variable and $g_2(t)$, which is helpful for understanding why the degrees of freedom rapidly increased.   When $\|\hat{ \bm{\beta}}(t)\|$ attained $343.76$,  the sign of $g_2(t)$ changed.  At this point, $\hat{\beta}_2(t)>0$. Thus,  $g_2(t)\hat{\beta}_2(t)<0$ at $\|\hat{ \bm{\beta}}(t)\|=343.76$, which means $\hat{\beta}_2(t)$ was updated because the set $S$ in  line 4 in Algorithm \ref{GPSalgorithm for GDF modified} became $S$=\{2\}.  When $|g(t)|$ was sufficiently small, $m$ in  (\ref{GPS_update_modification3}) became very large, which made a substantial change  in the degrees of freedom.    


The estimated standardized coefficients for the diabetes data based on the lasso, elastic net ($\alpha=0.5$) and generalized elastic net ($\alpha=0.5$) are reported in Table \ref{table:coef}.  The tuning parameter was selected by $C_p$ criterion, where the degrees of freedom was computed via the proposed procedure.  The generalized elastic net yielded the sparsest solution.   On the other hand, the elastic net did not produce sparse solution.

\begin{figure}[t]
      \centering
	\includegraphics[width=130mm]{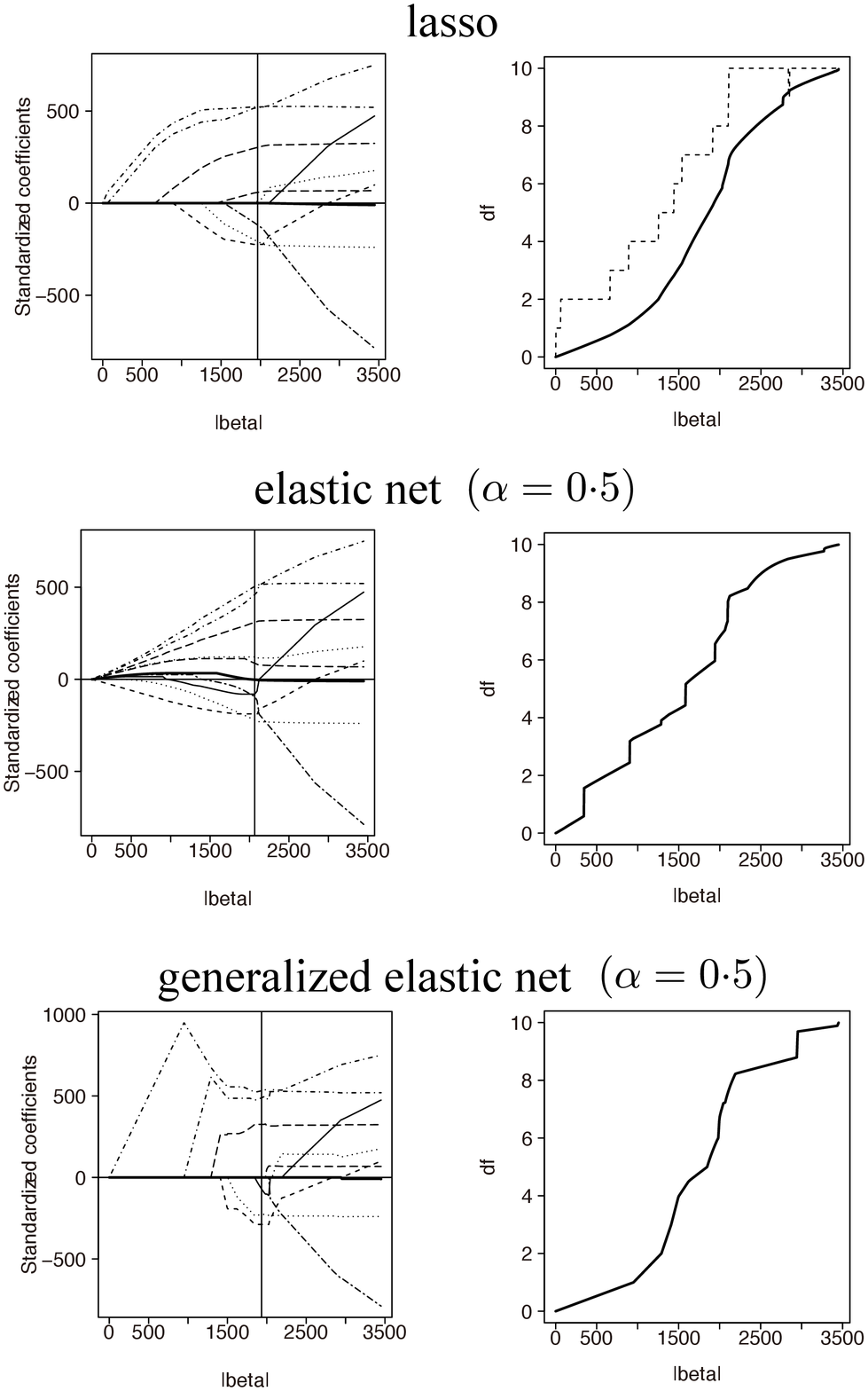}
 \caption{The solution path (left panel) and the degrees of freedom (right panel). The vertical line in the solution path indicates the selected model by  $C_p$ criterion.  The solid line of upper right panel is the degrees of freedom of the lasso \citep{Zouetal:2007}}
      \label{fig:path_df_diabetes} 
   \end{figure}

\begin{table}[t]
\caption{The estimated standardized coefficients for the diabetes data based on the lasso, elastic net ($\alpha=0.5$) and generalized elastic net ($\alpha=0.5$). }
\label{table:coef}
\begin{center}
\begin{tabular}{rrrrrrrrrrrr}
  \hline
 & (Intercept) & age & sex & bmi & map & tc & ldl & hdl & tch & ltg & glu \\ 
  \hline
lasso& 152 & 0 & $-$209 & 522 & 303 & $-$120 & 0 & $-$224 & 12 & 518 & 58 \\ 
enet & 152 & $-$2 & $-$220 & 504 & 309 & $-$93 & $-$81 & $-$188 & 122 & 460 & 87 \\ 
genet & 152 & 0 & $-$228 & 532 & 326 & 0 & $-$70 & $-$288 & 0 & 489 & 0 \\ 
   \hline
\end{tabular}
\end{center}
\end{table}

\begin{figure}[p]
      \centering
	\includegraphics[width=74mm]{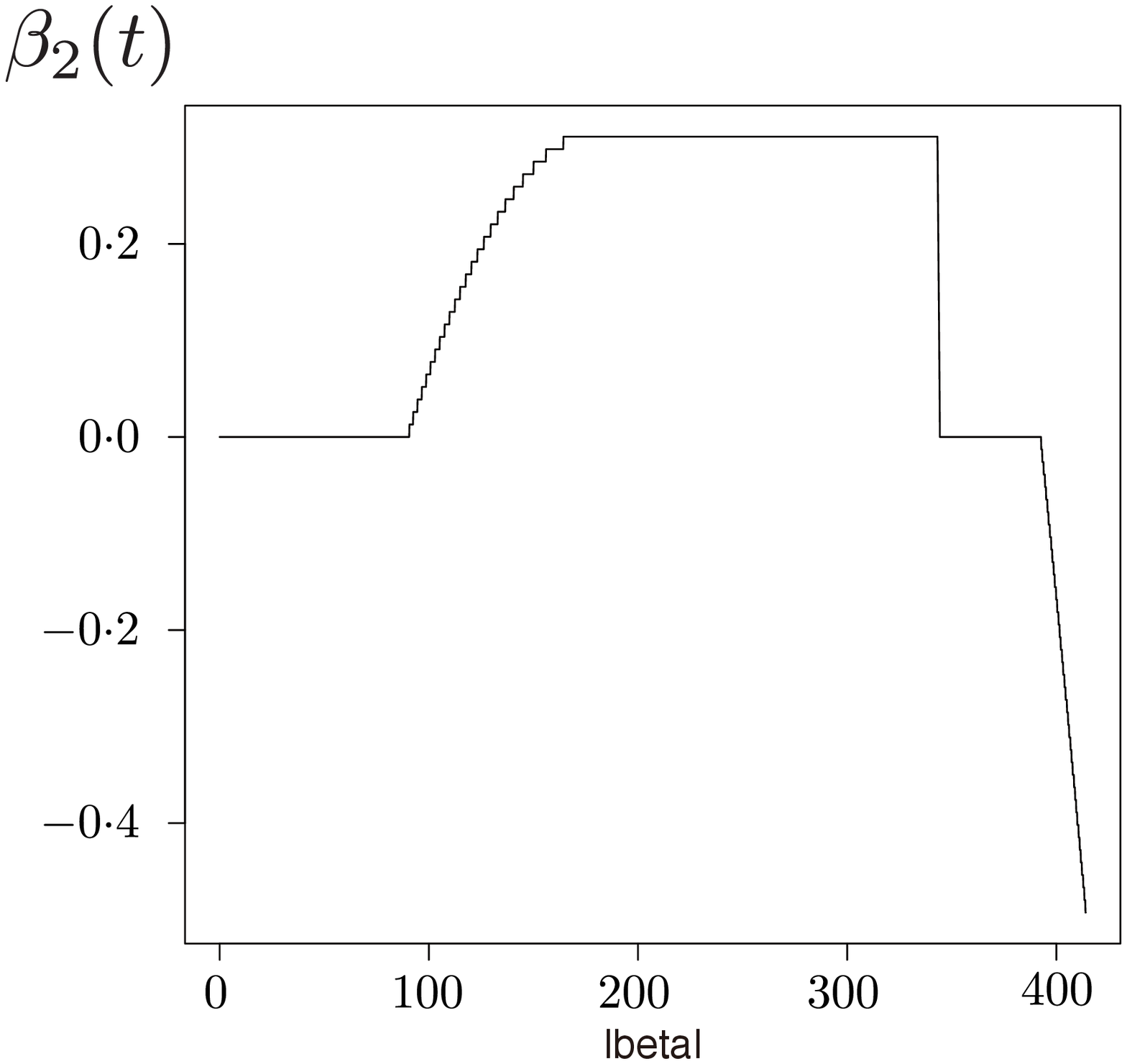}
	\includegraphics[width=70mm]{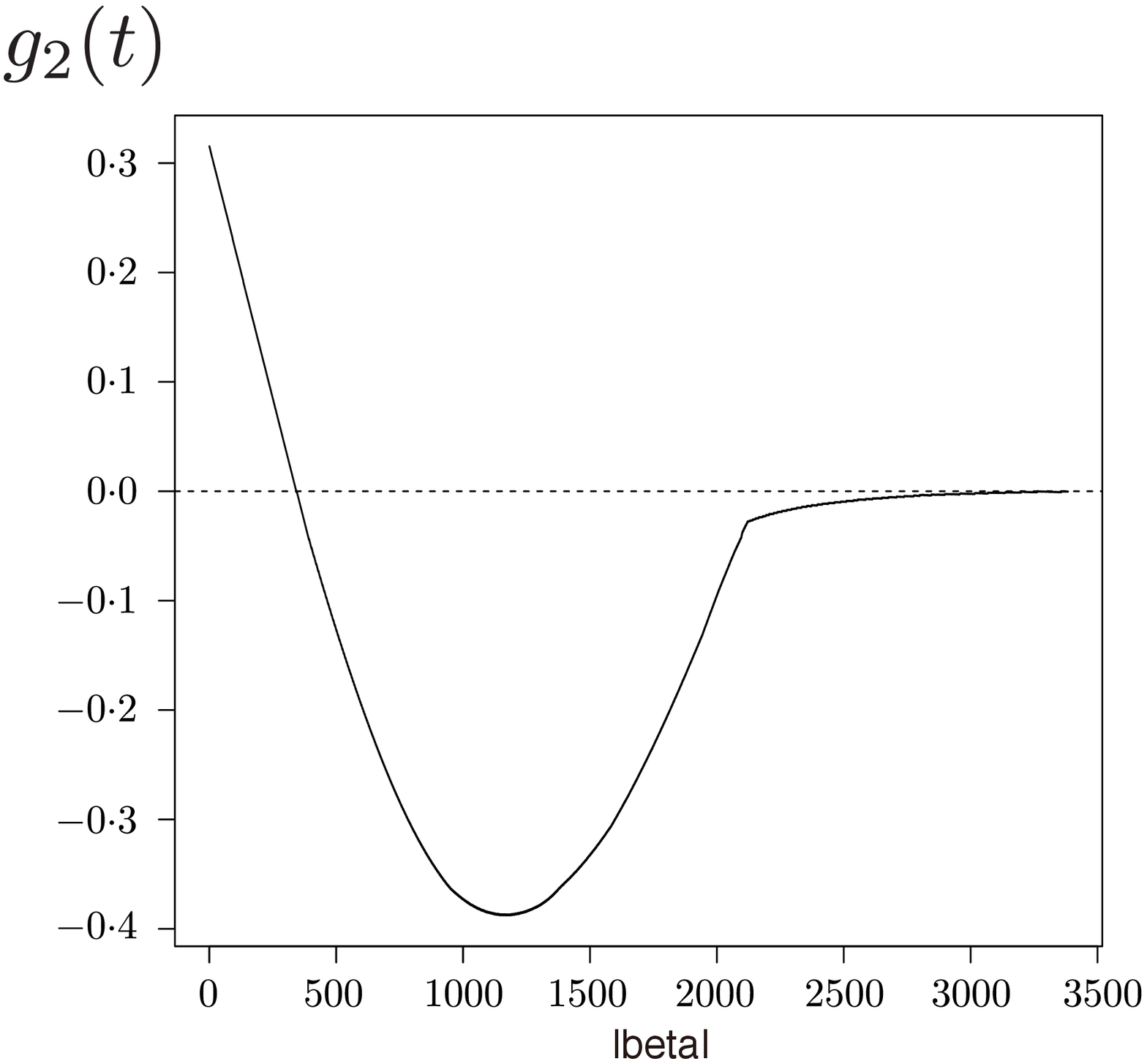}
 \caption{The solution path $\beta_2(t)$ (left panel) 
 and $g_2(t)$ (right panel) of the 2nd variable for elastic net with $\alpha=0.5$.  } \label{fig:2nd_df_diabetes} 
   \end{figure}
\clearpage
\section{Concluding Remarks}
We have proposed a new procedure for selecting tuning parameters in sparse regression modeling
via regularization, for which the degrees of freedom was calculated  by a computationally-efficient algorithm.   
Our procedure can be applied to construct model selection criteria for evaluating models estimated by the regularization methods with a wide variety of convex and non-convex penalties.      
Monte Carlo simulations were conducted to investigate  the effectiveness of the proposed procedure.   
Although the cross validation has been widely used to select tuning parameters, the model selection criteria based on the degrees of freedom often yielded better results, especially for non-convex penalties such as generalized elastic net.  


In the present paper, we considered a computationally-efficient algorithm to select the tuning parameter in the sparse regression model.   For more general models including generalized linear models, multivariate analysis such as factor analysis and graphical models, it is also important to select appropriate values of tuning parameters.  As a future research topic, it is interesting to introduce a new selection algorithm that handles large models by unifying the mathematical approach and computational algorithms.
\renewcommand{\theequation}{A\arabic{equation}}
\setcounter{equation}{0}
\section*{Appendix}
\subsection*{Proof of  Lemma \ref{Lem_cp}}
First, we divide $\|  \bm{y}^{\rm new} - \hat{ \bm{\mu}} \|^2$ into three terms as follows:
\begin{eqnarray}
\|  \bm{y}^{\rm new} - \hat{ \bm{\mu}} \|^2 &=& \|  \bm{y}^{\rm new} - { \bm{\mu}} \|^2 + \| \bm{\mu} - \hat{ \bm{\mu}} \|^2  +2(  \bm{y}^{\rm new} -  \bm{\mu}  )^T( \bm{\mu} - \hat{ \bm{\mu}}). \label{CP_der1}
\end{eqnarray}
The term $\|  \bm{\mu} - \hat{ \bm{\mu}} \|^2$ is expressed as
\begin{eqnarray}
\| \bm{\mu} - \hat{ \bm{\mu}} \|^2 &=&  \| \bm{y} - \hat{ \bm{\mu}} \|^2 +\| \bm{y} - { \bm{\mu}} \|^2 - 2( \bm{y} - \bm{\mu})^T( \bm{y} - \hat{ \bm{\mu}})\cr
&=&  \| \bm{y} - \hat{ \bm{\mu}} \|^2 - \| \bm{y} - { \bm{\mu}} \|^2 + 2( \bm{y} - \bm{\mu})^T(\hat{ \bm{\mu}} - \bm{\mu})\cr
&=&  \| \bm{y} - \hat{ \bm{\mu}} \|^2 - \| \bm{y} - { \bm{\mu}} \|^2 + 2( \bm{y} - \bm{\mu})^T(\hat{ \bm{\mu}} -E_{ \bm{y}}[\hat{ \bm{\mu}}]) \cr
&&+ 2( \bm{y} - \bm{\mu})^T(E_{ \bm{y}}[\hat{ \bm{\mu}}]- \bm{\mu})
.\label{CP_der2}
\end{eqnarray}
Substituting (\ref{CP_der2}) into (\ref{CP_der1}) gives us
\begin{eqnarray*}
\|  \bm{y}^{\rm new} - \hat{ \bm{\mu}} \|^2 &=& \|  \bm{y}^{\rm new} - { \bm{\mu}} \|^2 + \| \bm{y} - \hat{ \bm{\mu}} \|^2 - \| \bm{y} - { \bm{\mu}} \|^2 + 2( \bm{y} - \bm{\mu})^T(\hat{ \bm{\mu}} -E_{ \bm{y}}[\hat{ \bm{\mu}}]) \cr
&& + 2( \bm{y} - \bm{\mu})^T(E_{ \bm{y}}[\hat{ \bm{\mu}}]- \bm{\mu})+2(  \bm{y}^{\rm new} -  \bm{\mu}  )^T( \bm{\mu} - \hat{ \bm{\mu}}) .
\end{eqnarray*}
By taking expectation of $E_{ \bm{y}^{\rm new}}$,  we obtain
\begin{eqnarray*}
E_{ \bm{y}^{\rm new}} [\|  \bm{y}^{\rm new} - \hat{ \bm{\mu}} \|^2] &=& E_{ \bm{y}^{\rm new}}[\|  \bm{y}^{\rm new} - { \bm{\mu}} \|^2] + \| \bm{y} - \hat{ \bm{\mu}} \|^2 - \| \bm{y} - { \bm{\mu}} \|^2 \cr
&&+2( \bm{y} - \bm{\mu})^T(\hat{ \bm{\mu}} -E_{ \bm{y}}[\hat{ \bm{\mu}}]) + 2( \bm{y} - \bm{\mu})^T(E_{ \bm{y}}[\hat{ \bm{\mu}}]- \bm{\mu}).
\end{eqnarray*}
Equation (\ref{Mallows_der}) can be derived by taking the expectation of $E_{ \bm{y}}$ and using $E_{ \bm{y}^{\rm new}}[\|  \bm{y}^{\rm new} - { \bm{\mu}} \|^2] = E_{ \bm{y}}[\|  \bm{y} - { \bm{\mu}} \|^2] =n\tau^2$.
\subsection*{Derivation of generalized path seeking algorithm}
We derive the generalized path seeking algorithm.  First, the following lemma is provided.
\begin{lemma}\label{Lem_GPS_der_1}
Let us consider the following problem.
\begin{equation}
\Delta \hat{ \bm{\beta}}(t) =  \underset{\Delta { \bm{\beta}}}{\operatorname{argmin}}[R(\hat{ \bm{\beta}}(t) +  \Delta  \bm{\beta}) - R(\hat{ \bm{\beta}}(t)) ] \quad {\rm s.t.} \quad P(\hat{ \bm{\beta}}(t) +  \Delta  \bm{\beta}) - P(\hat{ \bm{\beta}}(t))  \le \Delta t. \label{problem_deltat}
\end{equation}
The solution is $\Delta \hat{ \bm{\beta}}(t) = \hat{ \bm{\beta}}(t+\Delta t) - \hat{ \bm{\beta}}(t)$.
\end{lemma}

\begin{proof}
The constraint in (\ref{problem_deltat}) is written as
\begin{equation*}
P(\hat{ \bm{\beta}}(t) +  \Delta  \bm{\beta}) \le  P(\hat{ \bm{\beta}}(t))  + \Delta t \le t+\Delta t .
\end{equation*}
Assume that $ \bm{\beta}^* = \hat{ \bm{\beta}}(t) + \Delta  \bm{\beta}$.  Then, the problem (\ref{problem_deltat}) is   
\begin{equation*}
\hat{ \bm{\beta}}^* = \underset{{{ \bm{\beta}}^* }}{\operatorname{argmin}} \ [R( { \bm{\beta}}^*) ] \quad {\rm s.t.} \quad P( { \bm{\beta}}^*) \le t + \Delta t.
\end{equation*}
The solution is $\hat{ \bm{\beta}}^* = \hat{ \bm{\beta}}(t+\Delta t )$, which leads to $\Delta \hat{ \bm{\beta}}(t) = \hat{ \bm{\beta}}(t+\Delta t ) - \hat{ \bm{\beta}}(t)$.
\end{proof}

When $\Delta t $ is sufficiently small, the problem (\ref{problem_deltat}) can be approximately written as  
\begin{eqnarray}
&&\Delta \hat{ \bm{\beta}}(t) = \underset{{\{\Delta \beta_j \mid j=1,\dots,p\} }}{\operatorname{argmax}}  
\sum_{j=1}^p g_j(t) \cdot \Delta \beta_j \quad  \\
&&{\rm s.t.} \quad \sum_{\hat{\beta}_j(t)=0}p_j(t) \cdot | \Delta \beta_j | +  \sum_{\hat{\beta}_j(t) \ne 0}p_j(t) \cdot sign(\hat{\beta}_j(t)) \cdot  \Delta \beta_j   \le \Delta t. \label{condition_adj_GPS}
\end{eqnarray}
Since all coefficient paths $\{ \hat{\beta}_j(t)  \mid {j=1,\dots,p} \}$ are monotone functions of $t$, we have $\{ sign(\hat{\beta}_j(t))= sign(\Delta \hat{\beta}_j(t)) \mid {j=1,\dots,p} \}$. Therefore, the problem in (\ref{condition_adj_GPS}) can be expressed as
\begin{equation}
\Delta \hat{ \bm{\beta}}(t) = \underset{{\{\Delta \beta_j  \mid j=1,\dots,p \} }}{\operatorname{argmax}} 
\sum_{j=1}^p g_j(t) \cdot \Delta \beta_j \quad  {\rm s.t.} \quad \sum_{j=1}^p p_j(t) \cdot | \Delta \beta_j |   \le \Delta t. \label{Problem_adj}
 \end{equation} 
 The problem in (\ref{Problem_adj}) can be viewed as a linear programming.  Then, the updates in (\ref{addall})  and (\ref{addone}) can be derived.  
\bibliographystyle{asa} 
\bibliography{paper-ref} 

\end{document}